\documentclass[3p,12pt,sort&compress]{elsarticle}
\pdfoutput=1

\usepackage{amsmath}
\usepackage{amsthm}

\usepackage{lmodern}
\usepackage[T1]{fontenc}
\usepackage{amsfonts}
\usepackage{amssymb}
\DeclareMathAlphabet{\mathbfi}{OML}{cmm}{b}{it}
\usepackage{mathrsfs}

\let\originalleft\left
\let\originalright\right
\renewcommand{\left}{\mathopen{}\mathclose\bgroup\originalleft}
\renewcommand{\right}{\aftergroup\egroup\originalright}

\allowdisplaybreaks

\usepackage{xcolor}
\usepackage[colorlinks=true,linktocpage=true,pdfa=true]{hyperref}

\newcommand{\mathe}{\mathrm{e}}
\newcommand{\mathi}{\mathrm{i}}
\newcommand{\total}{\mathop{}\!\mathrm{d}}
\newcommand{\eqend}[1]{\,#1}
\newcommand{\man}{\mathcal{M}}

\newtheorem{theorem}{\textsc{Theorem}}[section]
\newtheorem{lemma}[theorem]{\textsc{Lemma}}
\newtheorem{proposition}[theorem]{\textsc{Proposition}}
\newtheorem{corollary}[theorem]{\textsc{Corollary}}
\newtheorem{example}[theorem]{\textsc{Example}}

\journal{Journal of Geometry and Physics}

\begin{document}

\title{Global Hyperbolicity and Self-adjointness}

\author[1]{Markus B. Fr{\"o}b}
\ead{mfroeb@itp.uni-leipzig.de}
\author[1]{Albert Much}
\ead{much@itp.uni-leipzig.de}
\author[2]{Kyriakos Papadopoulos\corref{cor1}}
\ead{kyriakos.papadopoulos@ku.edu.kw}

\affiliation[1]{organization={Institut f{\"u}r Theoretische Physik, Universit{\"a}t Leipzig}, addressline={Br{\"u}derstra{\ss}e 16}, postcode={04103}, city={Leipzig}, country={Germany}}
\affiliation[2]{organization={Department of Mathematics, Faculty of Science, Kuwait University}, postcode={Safat 13060}, city={Sabah Al Salem University City}, country={Kuwait}}

\cortext[cor1]{Corresponding author}

\begin{abstract}
We show that the spatial part of the Klein--Gordon operator is an essentially self-adjoint operator on the Cauchy surfaces of various classes of spacetimes. Our proof employs the intricate connection between global hyperbolicity and geodesically complete Riemannian surfaces, and concludes by proving global hyperbolicity of the spacetimes under study.
\end{abstract}


\maketitle

\tableofcontents
\clearpage

\section{Introduction}

Global hyperbolic spacetimes are important in the study of general relativity and quantum field theory due to their well-defined and robust causal structure. This causal structure ensures the stability and predictability of physical processes within the spacetime, namely it ensures the existence of a well-defined Cauchy problem. Global hyperbolicity is the most strict form of causal structure in a spacetime~\cite{Minguzzi:2019mbe}. A spacetime $(\man, g)$ (which we take to be Hausdorff, second countable, connected, time-oriented and with $C^p$ metric $g$ for $p \geq 2$) is said to be globally hyperbolic if it is causal (there exist no closed causal curves) and for every pair of points $p$ and $q$ in $\man$, the causal diamonds $J^-(p) \cap J^+(q)$ are compact. The definition of global hyperbolicity has undergone a number of changes in the last decades, with one of the first results being that the condition of strong causality in the original definition~\cite{hawkingellis} can be replaced by causality~\cite{Bernal:2006xf}. Moreover, it is now known that if the dimension of $\man$ is at least 3 and $\man$ is non-compact, the causality condition follows in fact from the compactness of the causal diamonds~\cite{Hounnonkpe:2019rel}. Another equivalent definition of a globally hyperbolic spacetime is the existence of a Cauchy surface~\cite{Geroch:1970uw}. A Cauchy surface $\Sigma\subset \man$ for $(\man,g)$ is a closed set such that no timelike curve intersects $\Sigma$ more than once, and whose domain of dependence (the set of all points $p$ in $\man$ such that every inextendible causal curve through $p$ intersects $\Sigma$) is $\man$.

The non-existence of closed timelike curves in a global hyperbolic spacetime ensures that the initial-value problem in classical field theory is well-defined. That is, specifying initial data on a Cauchy surface uniquely determines the future and past evolution of the fields. This then ensures that also the construction of quantum fields is free of pathologies, and we refer the interested reader to the reviews~\cite{Hollands:2014eia,Fredenhagen:2014lda} and references therein. On the other hand, while it is sometimes possible to define quantum fields also on non-globally hyperbolic spacetimes, various pathologies can arise in these cases (see for example Refs.~\cite{Wald:1980jn,Kay:1992es,Ishibashi:2003jd,Seggev:2003rp}).

In this context, one of the main aspects is the connection between global hyperbolicity and the self-adjointness of certain differential operators that are derived from the equations of motion of the theory. In particular, for the simplest quantum field theory, the theory of a linear Hermitean scalar field, the relevant differential operator is the spatial part of the Klein--Gordon operator. For a correct formulation of the quantum theory even in this simple case, it needs to be essentially self-adjoint~\cite{Ashtekar:1975zn,Kay:1978yp}, see also the recent works~\cite{derezinskisiemssen2019,gerardwrochna2019,dappiaggidragoferreira2019,nakamurataira2023} and references therein for related approaches and new results. Namely, since the initial-value problem is well-defined, the full Hilbert space of the quantum theory can be taken to be the Fock space $\mathcal{F}$ constructed over the Hilbert space $\mathcal{H}$ of initial-value data. In turn, $\mathcal{H}$ is constructed from the classical symplectic space of initial-value data with the canonical symplectic form $\sigma$ and a compatible complex structure $J$, which define the scalar product $\mu$ of $\mathcal{H}$ according to $\mu(f, g) = \sigma(f, J g)$. For a static spacetime, using the well-known decomposition of solutions $\phi$ of the Klein--Gordon equation into positive and negative frequency parts $\phi^\pm$, the corresponding complex structure is simply $J \phi = \mathi \phi^+ - \mathi \phi^-$, such that the formulation with a general $J$ generalizes the construction of the quantum theory to arbitrary globally hyperbolic spacetimes. Of course, in the general case $J$ will be time-dependent. In turn, given a complex structure $J$ and a timelike curve with tangent vector $t^\mu$, the corresponding Schr{\"o}dinger equation reads $H \phi = - J (\mathcal{L}_t \phi)$, where $\mathcal{L}$ denotes the Lie derivative. In fact, this equation defines the one-particle Hamiltonian that generates translations along the integral curves of $t^\mu$, and it has been shown that~\cite[Appendix]{Ashtekar:1975zn}, \cite[Thm.~4.1]{Kay:1978yp}
\begin{enumerate}
\item[a)] for static and stationary spacetimes there exists a unique complex structure $J$ such that the expectation value of $H$ in a one-particle state is equal to the classical energy of the solution with the same initial-value data, and
\item[b)] in the static case and taking $t^\mu$ to be the static Killing vector field, $J$ is the complex structure determined by the decomposition into positive and negative frequency parts, and $H$ is the usual Hamiltonian.
\end{enumerate}

For a general spacetime, the construction of a suitable complex structure $J$ is quite involved. It is here that the self-adjointness of the spatial part of the Klein--Gordon operator comes into play, since one can define $J$ from its spectral decomposition~\cite[Thm.~7.1]{Kay:1978yp}, \cite{Brum:2013bia,Much:2018ehc}, \cite[Sec.~7.2]{Corichi:2002qd}. Furthermore, the self-adjointness of the full Klein--Gordon operator is essential for the construction of Feynman propagators in general globally hyperbolic spacetimes~\cite{derezinskisiemssen2019,Siemssen:2019qlb}, for which in turn a key assumption is the self-adjointness of the spatial part of the Klein--Gordon operator~\cite[Assump.~1]{derezinskisiemssen2019}. Of course, for a linear theory both approaches are completely equivalent since one can construct a complex structure from a Feynman propagator and vice versa, hence we see that self-adjointness of the spatial part of the Klein--Gordon operator is essential for a correct formulation of the quantum theory.

Fortunately, many spacetimes of physical importance are globally hyperbolic, including Minkowski spacetime, the exterior region of the Schwarzschild black hole spacetime, and the Friedmann--Lema{\^\i}tre--Robertson--Walker (FLRW) spacetimes relevant for cosmology. In this paper, we show that for all these spacetimes and several others, the spatial part of the Klein-Gordon operator is essentially self-adjoint. Our proof connects existing results in the literature, linking global hyperbolicity to geodesic completeness of the Cauchy hypersurfaces and finally the self-adjointness of the spatial part of the Klein--Gordon operator.

In addition to providing the proof of essential self-adjointness, our work has further applications. Collecting and integrating various theorems scattered throughout the literature, we demonstrate that specific \emph{classes} of spacetimes are globally hyperbolic without requiring the explicit construction of Cauchy surfaces or proving the compactness of the intersection of causal past and future sets. Since these proofs are often quite technical, our work provides a simplification of the proofs of global hyperbolicity for these classes of spacetimes. Moreover, our results can be used in the context of the semiclassical Einstein equations
\begin{equation}
\label{see}
R_{\mu\nu} - \frac{1}{2} R g_{\mu\nu} = 8 \pi G_\text{N} \, \omega\Bigl( T^\text{ren}_{\mu\nu} \Bigr) \eqend{,}
\end{equation}
where $R_{\mu\nu}$ is the Ricci tensor of the spacetime, $G_\text{N}$ is Newton's constant, $\omega$ is the quantum state (constructed for a linear theory from the scalar product $\mu$ of $\mathcal{H}$), and $T^\text{ren}_{\mu\nu}$ is the renormalized stress tensor in the quantum theory. Eq.~\eqref{see} can be seen as a first step towards a theory of quantum gravity, where one incorporates the back-reaction of quantum fields on the geometry. However, it is very difficult to solve in general: one needs to find a globally hyperbolic spacetime $(\man,g)$ and quantum state $\omega$ which itself depends on $g$, such that the equation is fulfilled self-consistently. We refer the interested reader to the works~\cite{Hack:2010iw,Pinamonti:2013wya,Gottschalk:2018kqt,Meda:2020smb,Meda:2021zdw} and references therein for various mathematically rigorous approaches to this problem. Knowing that a large class of spacetimes is globally hyperbolic can help in solving the semiclassical Einstein equations~\eqref{see} or in showing that solutions exist, since for a given class of spacetimes they reduce to differential equations for the arbitrary functions on which the metric depends.

This article is structured as follows: In Sec.~\ref{sec:kg}, we review the connection between global hyperbolicity, the completeness of the Cauchy surfaces and the self-adjointness of the spatial part of the Klein--Gordon operator, combining the results of Refs.~\cite{Bernal:2003jb,Bernal:2004gm,Much:2018ikx}. In Sec.~\ref{sec:global} we first collect known results regarding global hyperbolicity and complete Riemannian manifolds, both general and for certain classes of spacetimes. Building on those, we present new results on complete Riemannian manifolds (Props.~\ref{prop:bounded}, \ref{prop:warpedsphere}, \ref{prop:warpedhom}) and the global hyperbolicity of certain warped product manifolds (Props.~\ref{prop:nonstatic-twodim} and~\ref{prop:nonstatic-fourdim}). We also present a proof for the known and widely used result of Prop.~\ref{prop:onedim-arclength}, where we haven't been able to locate an existing proof. We then apply the general theorems to explicit examples: wormhole spacetimes in Sec.~\ref{sec:wormhole}, spherically symmetric, static spacetimes in Sec.~\ref{sec:spherical} and non-static spacetimes in Sec.~\ref{sec:nonstatic}. Some of these, such as the exterior Schwarzschild and Reissner--Nordström spacetimes, the de~Sitter spacetime or the FLRW spacetimes, are known to be globally hyperbolic, but our treatment has the advantage that the proofs are a simple corollary of the results of Sec.~\ref{sec:global}. To the best of our knowledge, our results for the global hyperbolicity (and the connection to the essential self-adjointness of the spatial part of the Klein--Gordon operator) of the other examples, namely the Morris--Thorne wormhole spacetime and its generalizations, topological black holes, and type I and IX Bianchi universes, are new.

\section{The spatial Klein--Gordon Equation}
\label{sec:kg}

Globally hyperbolic spacetimes admit a global time function, making them topologically equivalent to $\man \sim \mathbb{R} \times \Sigma$ for a Cauchy surface $\Sigma$~\cite{Geroch:1970uw}. In~\cite{Bernal:2003jb}, Bernal and S{\'a}nchez resolved a long-standing conjecture by proving that every globally hyperbolic spacetime possesses a \emph{smooth} foliation into Cauchy surfaces~\cite[Theorem 1.1]{Bernal:2003jb}. Furthermore, the metric in such a foliation takes on a particular form~\cite[Theorem 1.1]{Bernal:2004gm}, 
\begin{equation}
\label{eq:foliation}
g = - N^2 \total t^2 + h_{ij} \total x^i \total x^j \eqend{,}
\end{equation}
where $\Sigma$ is a smooth 3-manifold, $t \colon \mathbb{R} \times \Sigma\mapsto \mathbb{R}$ is the natural projection on the first factor, $N \colon \mathbb{R} \times \Sigma \mapsto (0,\infty)$ (the lapse) is a smooth function, and $h$ (the spatial metric) is a symmetric tensor field on $\mathbb{R} \times \Sigma$. Furthermore, each hypersurface $\Sigma_t = \{t\} \times \Sigma \subset \man$ at constant $t$ is a Cauchy surface, and the restriction $h(t)$ of $h$ to $\Sigma_t$ is a Riemannian metric (which in particular implies that $\Sigma_t$ is spacelike). For such manifolds, the Klein--Gordon equation $\left( \nabla^2 - V \right) \phi = 0$ with an external potential $V = V(t,x)$ takes the form
\begin{equation}
\Bigl[ \partial_t^2 + f(t,x) \partial_t + w^2 \Bigr] \phi = 0 \eqend{,}
\end{equation}
where $f(t,x) = - \partial_t \ln N + \frac{1}{2} \partial_t \ln \det h$ and where the spatial part of the Klein-Gordon equation $w^2$ is given by the operator
\begin{equation}
\label{op}
\begin{split}
w^2 &= - \frac{N}{\sqrt{\det h}} \partial_i \Bigl( \sqrt{\det h} N h^{ij} \partial_j \Bigr) + N^2 \, V \\
&= - N^2 \Bigl( \Delta_h - V \Bigr) - N h^{ij} \partial_i N \partial_j \eqend{,}
\end{split}
\end{equation}
where $\Delta_h$ is the Laplace--Beltrami operator associated to the spatial metric $h$. According to~\cite[Remark 3.1]{Much:2018ikx}, $- N^2 \Delta_h$ is a symmetric, densely defined and  positive operator on the weighted Hibert space $L^2(\Sigma,\tilde{\mu})$ with measure $\total \tilde{\mu} = N^{-1} \sqrt{\det h} \total^3 x$ and scalar product
\begin{equation}
\langle \phi, \psi \rangle_{L^2(\Sigma,\tilde{\mu})} = \int_\Sigma \overline{\phi(x)} \psi(x) N^{-1} \sqrt{\det h} \total^3 x \eqend{.}
\end{equation}
Taking into account that $- N^2 \Delta_h \geq 0$, \cite[Theorem 4.1]{Much:2018ikx} then shows that $w^2$ is essentially self-adjoint on the dense subspace $C_0^\infty(\Sigma) \subset L^2(\Sigma,\tilde{\mu})$ if the Riemannian manifold $(\Sigma, \tilde{h} = N^{-2} h)$ is geodesically complete for each fixed $t$, and if the rescaled potential $N^2 V$ is semi-bounded from below and locally $L^2$ for each fixed $t$.

A usual choice for the potential is $V = m^2 + \xi R$, where $m \geq 0$ is the mass of the scalar field, $\xi \in \mathbb{R}$ and $R$ the Ricci scalar of the spacetime. $V$ is then semi-bounded from below if either $\xi \geq 0$ and $R \geq -c$ or $\xi \leq 0$ and $R \leq c$ for some constant $c > 0$. In particular, this always includes the cases of minimal coupling $\xi = 0$ or vacuum spacetimes with $R = 0$. In the following, we assume that the conditions on $V$ are satisfied and investigate the geodesic completeness of $(\Sigma, \tilde{h})$ for various classes of spacetimes. By~\cite[Theorem 4.1]{Much:2018ikx}, this then implies in each case the essential self-adjointness of $w^2$, the spatial part of the Klein--Gordon operator $\nabla^2 - V$. In particular, we treat wormhole spacetimes in Sec.~\ref{sec:wormhole}, static, spherically symmetric spacetimes in Sec.~\ref{sec:spherical} and nonstatic spacetimes in Sec.~\ref{sec:nonstatic}.

Excluded from our analysis are general sliced spaces, which have the form~\eqref{eq:foliation} but with $\total x^i$ replaced by $\total x^i - N^i \total t$, where $N^i$ (the shift vector) is a smooth function. Conditions for the global hyperbolicity of sliced spaces were given in Refs.~\cite{Choquet-Bruhat:2002nwg,Cotsakis:2003tw}, and we also refer the reader to the review~\cite{Finster:2021pvx} for other recent advances on global hyperbolicity.\footnote{We note that a conjecture made in~\cite{Finster:2021pvx} was disproved in~\cite{Sanchez:2021ciy}.} Since the analogue of~\cite[Theorem 4.1]{Much:2018ikx} is only known for stationary spacetimes~\cite{Finster:2019deg} but not in general (to the best of our knowledge), we leave the study of essential self-adjointness of the spatial part of the Klein--Gordon operator in sliced spaces for future work.

\section{Theorems on global hyperbolicity}
\label{sec:global}

We first collect various existing theorems regarding global hyperbolicity. Since global hyperbolicity is determined by the causal structure, we have
\begin{proposition}[{\cite{Geroch:1970uw}}]
\label{prop:conformal}
Let $(\man,g)$ be a globally hyperbolic spacetime and $\Omega\colon \man \to (0,\infty)$ a smooth conformal factor. Then $(\man, \Omega g)$ is globally hyperbolic.
\end{proposition}
Many explicit spacetimes are given by warped products, for which we have
\begin{proposition}[{\cite[Thm.~3.66]{beem}, \cite[Lemma A.5.14]{baerginouxpfaeffle2007}}]
\label{prop:warpedproduct1}
Consider the warped product manifold $\man = (a,b) \times H$, $- \infty \leq a < b \leq \infty$, with metric $g = - \total t^2 + f h$, where $f \colon (a,b) \to (0,\infty)$ is a smooth positive function. Then $(\man,g)$ is globally hyperbolic if and only if $(H,h)$ is a complete Riemannian manifold.
\end{proposition}
This covers the case where the base manifold is one-dimensional. For a higher-dimensional base, we have
\begin{proposition}[{\cite[Thm.~3.68]{beem}}]
\label{prop:warpedproduct2}
Consider the warped product manifold $\man = M \times H$ with metric $g_\man = g + f h$, where $f \colon M \to (0,\infty)$ is a smooth positive function. Then $(\man,g_\man)$ is globally hyperbolic if and only if $(H,h)$ is a complete Riemannian manifold and $(M,g)$ is a globally hyperbolic manifold.
\end{proposition}
For a multiply warped product, we have an analogous result:
\begin{proposition}[{\cite[Thm.~3.3]{unal2000}}]
\label{prop:multiplywarpedproduct}
Consider the multiply warped product manifold $\man = M \times F_1 \times \cdots \times F_m$ with metric $g_\man = g + f_1 g_{F_1} + \cdots + f_m g_{F_m}$, where $f_i \colon M \to (0,\infty)$, $i \in \{1,\ldots,s\}$ are all smooth positive functions. Then $(\man,g_\man)$ is globally hyperbolic if and only if $(M,g)$ is globally hyperbolic and $(F_i,g_{F_i})$ is a complete Riemannian manifold for all $i \in \{1,\ldots,s\}$.
\end{proposition}

To effectively apply these propositions, we have to know also necessary and sufficient conditions for a Riemannian manifold to be complete.
\begin{lemma}[{\cite[Lemma~7.2]{bishoponeill1969}}]
\label{lemma:prod}
Consider the warped product manifold $M \times H$ with metric $g + f h$, where $f \colon M \to (0,\infty)$ is a smooth positive function. Then $(M \times H,g + f h)$ is complete if and only if both $(M,g)$ and $(H,h)$ are complete Riemannian manifolds.
\end{lemma}
Also doubly warped products have been studied, and for them we have
\begin{proposition}[{\cite[Prop.~3.1]{unal2001}}]
\label{prop:doublywarped}
Consider the doubly warped product manifold $B \times F$ with metric $g = f g_B + b g_F$, where $(B, g_B)$ and $(F, g_F)$ are complete Riemannian manifolds and $b \colon B \to (0,\infty)$, $f \colon F \to (0,\infty)$ are smooth functions. If $\inf b > 0$ or $\inf f > 0$, then $(B \times F, g)$ is complete.
\end{proposition}
In fact, even multiply warped products can be treated:
\begin{proposition}[{\cite[Thm.~4.14]{unal2000}}]
\label{prop:multiplywarped}
Consider the multiply warped product manifold $M = B \times F_1 \times \cdots \times F_m$ with metric $g = g_B + b_1 g_{F_1} + \cdots + b_m g_{F_m}$, where $b_i \colon B \to (0,\infty)$ for $i \in \{1,\ldots,m\}$ are smooth functions. If $(B, g_B)$ and $(F_i, g_{F_i})$ for $i \in \{1,\ldots,m\}$ are all complete Riemannian manifolds, then $(M,g)$ is a complete Riemannian manifold.
\end{proposition}

We can thus concentrate on the individual factors (which in our case will be one- and two-dimensional). We note that as a corollary of the Hopf--Rinow theorem, compact Riemannian manifolds are complete, which we will however not use in the following. Instead, we need
\begin{lemma}[{\cite[Lemma~5.4]{beem}}]
\label{lemma:hom}
If $(H,h)$ is a homogeneous Riemannian manifold, then $(H,h)$ is complete.
\end{lemma}
This covers in particular the flat case $(\mathbb{R}^d,\delta)$, the $d$-dimensional unit sphere $(\mathbb{S}^d, g_{\mathbb{S}^d})$, and the torus $(\mathbb{T}^d, \delta)$. For a subset of the real line, we obtain instead the following explicit condition:
\begin{proposition}
\label{prop:onedim-arclength}
The one-dimensional Riemannian manifold $(\man, g) = ( (a,b), \mathe^{2 \alpha(\cdot)})$ with $-\infty \leq a < b \leq \infty$ and a smooth function $\alpha$ is complete if and only if
\begin{equation*}
\int_a^p \mathe^{\alpha(r)} \total r = \infty \quad \text{and} \quad \int_p^b \mathe^{\alpha(r)} \total r = \infty
\end{equation*}
for an arbitrary point $p \in (a,b)$.
\end{proposition}
\begin{proof}
By the Hopf--Rinow theorem, it suffices to show that $(\man, g)$ is geodesically complete. We consider the affinely parametrized geodesic equation
\begin{equation*}
\gamma''(t) + \alpha'(\gamma(t)) [ \gamma'(t) ]^2 = 0
\end{equation*}
with initial data $\gamma(0) = q \in (a,b)$ and $\gamma'(0) = v \neq 0$. By a rescaling of $t$, we may assume without loss of generality that $v = 1$, and the geodesic equation can be solved in terms of
\begin{equation*}
t(\gamma) = \mathe^{- \alpha(q)} \int_q^\gamma \mathe^{\alpha(r)} \total r \eqend{.}
\end{equation*}
Geodesic completeness means that $\gamma(t)$ can be extended for all values $t \in \mathbb{R}$. We have $t(\gamma) > 0$ for $\gamma > q$, and can extend $t$ to $\infty$ if and only if $\int_q^b \mathe^{\alpha(r)} \total r = \infty$. Likewise, we have $t(\gamma) < 0$ for $\gamma < q$, and can extend $t$ to $-\infty$ if and only if $\int_a^q \mathe^{\alpha(r)} \total r = \infty$. Since $\alpha$ is a smooth function, the integral $\int_p^q \mathe^{\alpha(r)} \total r$ is finite for any $a < p, q < b$, and we conclude.
\end{proof}
Sometimes the concrete metric is too complicated to study exactly, but can be bounded from above and below. In this case, we have
\begin{proposition}
\label{prop:bounded}
Consider the Riemannian manifold $(M,g)$ and a metric $g_0$ such that $a g_0 \leq g \leq b g_0$ for some constants $a,b > 0$ (i.e., $g$ is sandwich bounded~\cite{Sanchez:2021ciy}). Then $(M,g)$ is a complete Riemannian manifold if and only if $(M,g_0)$ is a complete Riemannian manifold.
\end{proposition}
\begin{proof}
The condition on $g$ means that $g$ and $g_0$ induce equivalent norms, and hence $(M,g)$ and $(M,g_0)$ are either both complete or both incomplete.
\end{proof}

As applications of the above propositions and lemmata, we show
\begin{proposition}
\label{prop:warpedsphere}
Consider the Riemannian manifold $(H,h)$ with $H = (a,b) \times \mathbb{S}^2$, $- \infty \leq a < b \leq \infty$, and metric
\begin{equation*}
h = \mathe^{2 \alpha(r)} \total r^2 + \mathe^{2 \beta(r)} g_{\mathbb{S}^2} \eqend{,}
\end{equation*}
where $\alpha$ and $\beta$ are smooth functions and $g_{\mathbb{S}^2} = \total \vartheta^2 + \sin^2 \vartheta \total \varphi^2$ is the standard metric on the unit sphere with $\vartheta \in [0,\pi)$ and $\varphi \in [0,2\pi)$. $(H,h)$ is complete if and only if
\begin{equation*}
\int_a^p \mathe^{\alpha(r)} \total r = \infty \quad \text{and} \quad \int_p^b \mathe^{\alpha(r)} \total r = \infty
\end{equation*}
for an arbitrary point $p \in (a,b)$.
\end{proposition}
\begin{proof}
Taking $f = \mathe^{2 \beta} > 0$ in Lemma~\ref{lemma:prod}, we obtain that $(H,h)$ is complete if and only if $( (a,b), \mathe^{2\alpha(\cdot)} )$ and $(\mathbb{S}^2, g_{\mathbb{S}^2})$ are complete Riemannian manifolds. For the sphere, this follows from Lemma~\ref{lemma:hom}, while for the first factor we use Prop.~\ref{prop:onedim-arclength} to conclude. 
\end{proof}
Note that we can perform a change of coordinates to $\rho = \exp\left[ \int \mathe^{\alpha(r) - \beta(r)} \total r \right]$, which results in
\begin{equation}
h = \frac{\mathe^{2 \beta(r)}}{\rho^2} \left( \total \rho^2 + \rho^2 g_{\mathbb{S}^2} \right) \eqend{.}
\end{equation}
Assuming that $\alpha$ and $\beta$ are such that it is possible to choose $\rho \in [0,\infty)$, we see that in this case $(H,h)$ is conformally related to flat space, with the conformal factor depending on the distance from the origin. In fact, we even have
\begin{proposition}
\label{prop:warpedhom}
Consider the Riemannian manifold $(H,h)$ with $H = (a,b) \times K$, $- \infty \leq a < b \leq \infty$, and metric
\begin{equation*}
h = \mathe^{2 \alpha(r)} \total r^2 + \mathe^{2 \beta(r)} g_K \eqend{,}
\end{equation*}
where $\alpha$ and $\beta$ are smooth functions and $(K, g_K)$ is a homogeneous Riemannian manifold. $(H,h)$ is complete if and only if
\begin{equation*}
\int_a^p \mathe^{\alpha(r)} \total r = \infty \quad \text{and} \quad \int_p^b \mathe^{\alpha(r)} \total r = \infty
\end{equation*}
for an arbitrary point $p \in (a,b)$.
\end{proposition}
\begin{proof}
The proof proceeds exactly as the one of Prop.~\ref{prop:warpedsphere}, since Lemma~\ref{lemma:hom} applies to a general homogeneous Riemannian manifold.
\end{proof}

For doubly warped product spacetimes, one obtains simple corollaries of the above propositions and lemmata:
\begin{corollary}[{\cite{beempowell1982}}]
\label{corr:doublywarpedproduct1}
Consider the doubly warped product manifold $\man = M \times H$ with metric $g_\man = f g + e h$, where $f \colon H \to (0,\infty)$, $e \colon M \to (0,\infty)$ are smooth positive functions. Then $(\man,g_\man)$ is globally hyperbolic if and only if $(H,1/f h)$ is a complete Riemannian manifold and $(M,g)$ is a globally hyperbolic manifold.
\end{corollary}
\begin{proof}
By Prop.~\ref{prop:conformal} (taking $\Omega = 1/f$) it suffices to show that $(\man, g + e/f h)$ is globally hyperbolic, which in turn follows from Prop.~\ref{prop:warpedproduct2} (taking $f \to e$).
\end{proof}
A somewhat more complicated one is
\begin{corollary}[{\cite[Corr.~3.2]{unal2001}}]
\label{corr:doublywarpedproduct2}
Consider the doubly warped product manifold $\man = (a,b) \times B \times F$, $- \infty \leq a < b \leq \infty$, with metric $g = - \total t^2 + h(t) \bigl[ f g_B + b g_F \bigr]$, where $(B, g_B)$ and $(F, g_F)$ are complete Riemannian manifolds, and $h \colon (a,b) \to (0,\infty)$, $b \colon B \to (0,\infty)$, $f \colon F \to (0,\infty)$ are smooth positive functions. If $\inf b > 0$ or $\inf f > 0$, then $(\man,g)$ is globally hyperbolic.
\end{corollary}
\begin{proof}
By Prop.~\ref{prop:warpedproduct1} (with the replacement $f \to h$), the conclusion follows if $(B \times F, f g_B + b g_F)$ is a complete Riemannian manifold. This is seen to hold using Prop.~\ref{prop:doublywarped}.
\end{proof}

Finally, combining the above results we can show that
\begin{proposition}
\label{prop:nonstatic-twodim}
Consider the warped product manifold $\man = (a,b) \times \Sigma$, $-\infty \leq a < b \leq \infty$, where $\Sigma$ is a connected one-dimensional Riemannian manifold described by coordinates $x$, with metric $g = - \mathe^{h(t)} \total t^2 + \mathe^{k(t)} \total x^2$ and $h,k$ smooth functions. Then the spacetime $(\man, g)$ is globally hyperbolic.
\end{proposition}
\begin{proof}
By Prop.~\ref{prop:conformal} it suffices to show that the conformally related spacetime $(\man, \mathe^{-h(t)} g)$ is globally hyperbolic. This in turn follows from Prop.~\ref{prop:warpedproduct1} with $f \to \mathe^{k(\cdot) - h(\cdot)}$, since $(\Sigma,\total x^2)$ is complete by Prop.~\ref{lemma:hom}.
\end{proof}
We can iterate this construction by splitting off each time a Riemannian part:
\begin{proposition}
\label{prop:nonstatic-fourdim}
Consider the manifold $\man = (a,b) \times \mathbb{R}^3$, $-\infty \leq a < b \leq \infty$, described by coordinates $t,x,y,z$, with metric
\begin{equation*}
g = - \mathe^{h(t)} \total t^2 + \mathe^{k(t)} \total x^2 + \mathe^{l(t,x)} \total y^2 + \mathe^{m(t,x,y)} \total z^2
\end{equation*}
and $h,k,l,m$ smooth functions. Then the spacetime $(\man, g)$ is globally hyperbolic.
\end{proposition}
\begin{proof}
By Prop.~\ref{prop:warpedproduct2} with $M = (a,b) \times \mathbb{R}^2$, $H = \mathbb{R}$ and $f = m$, $(\man, g)$ is globally hyperbolic if $( (a,b) \times \mathbb{R}^2, - \mathe^{h(t)} \total t^2 + \mathe^{k(t)} \total x^2 + \mathe^{l(t,x)} \total y^2)$ is globally hyperbolic, since $(\mathbb{R}, \total z^2)$ is complete by Lemma~\ref{lemma:hom}. This in turn follows again from Prop.~\ref{prop:warpedproduct2}, now with $M = (a,b) \times \mathbb{R}$, $H = \mathbb{R}$ and $f = l$, since $(\mathbb{R}, \total y^2)$ is complete by Lemma~\ref{lemma:hom} and $( (a,b) \times \mathbb{R}, - \mathe^{h(t)} \total t^2 + \mathe^{k(t)} \total x^2)$ is globally hyperbolic by Prop.~\ref{prop:nonstatic-twodim}.
\end{proof}

\section{Wormhole Spacetimes}
\label{sec:wormhole}

We first consider the classic Morris--Thorne wormhole spacetime and a generalization. The Morris--Thorne wormhole spacetime $(\man,g_\text{WH})$, introduced in~\cite{Morris:1988cz} with the goal of using it for rapid interstellar travel, has the following metric: 
\begin{equation}
g_\text{WH} = - \total t^2 + \total r^2 + \left( b^2 + r^2 \right) g_{\mathbb{S}^2} \eqend{,}
\end{equation}
where $b > 0$ is the throat radius, $r \in \mathbb{R}$ is the proper radial coordinate, $t \in \mathbb{R}$, and the metric $g_{\mathbb{S}^2}$ on the unit sphere is given by $g_{\mathbb{S}^2} = \total\vartheta^2 + \sin^2 \vartheta \total \varphi^2$ with $\vartheta \in [0,\pi)$ and $\varphi \in [0,2\pi)$.

\begin{lemma}
\label{lemma:wh-globalhyperbolic}
The Morris--Thorne wormhole spacetime $(\man,g_\text{WH})$ is globally hyperbolic. 
\end{lemma}
\begin{proof}
We write the manifold $(\man,g)$ as a warped product manifold constructed from the two-dimensional Minkowski spacetime $(\mathbb{R}^2, \eta)$ as the base and the two-dimensional unit sphere $(\mathbb{S}^2, g_{\mathbb{S}^2})$ as the fiber, with $f$ the smooth positive function $f\colon \mathbb{R}^2 \to (0,\infty)$, $f(t,r) = b^2 + r^2$. By Prop.~\ref{prop:warpedproduct2}, $(\man,g_\text{WH})$ is globally hyperbolic if the base is globally hyperbolic and the fiber is a complete Riemannian manifold. The former holds by combining Prop.~\ref{prop:warpedproduct1} and Lemma~\ref{lemma:hom}, and the latter by Lemma~\ref{lemma:hom}.

Alternatively, using Prop.~\ref{prop:warpedproduct1}, $(\man,g_\text{WH})$ is globally hyperbolic if $(\mathbb{R} \times \mathbb{S}^2, h)$ with $h = \total r^2 + \left( b^2 + r^2 \right) g_{\mathbb{S}^2}$ is complete. This in turn follows from Prop.~\ref{prop:warpedsphere}, taking $a = - \infty$, $b = \infty$, $\alpha = 0$ and $\beta = \frac{1}{2} \ln\left( b^2 + r^2 \right)$.
\end{proof}

Regarding the Klein--Gordon operator on the Morris--Thorne wormhole spacetime, we obtain
\begin{proposition}
\label{prop:selfadjoint-wh}
The spatial part of the Klein--Gordon operator for the spacetime $(\man,g_\text{WH})$ (with a potential $V$ satisfying the conditions given in Sec.~\ref{sec:kg}) is essentially self-adjoint on $C_0^\infty(\mathbb{R} \times \mathbb{S}^2)$.
\end{proposition}
\begin{proof}
The metric of the Morris--Thorne wormhole has the form~\eqref{eq:foliation} with the lapse $N = 1$ and the spatial metric $h = \total r^2 + \left( b^2 + r^2 \right) g_{\mathbb{S}^2}$. Using~\cite[Theorem 4.1]{Much:2018ikx}, the conclusion thus follows if the Riemannian manifold $(\Sigma = \mathbb{R} \times \mathbb{S}^2, \tilde{h} = N^{-2} h)$ is geodesically complete. By the Hopf--Rinow theorem, this holds if $(\Sigma, h)$ is complete, which we have shown in the alternative proof of Lemma~\ref{lemma:wh-globalhyperbolic}.
\end{proof}

Let us generalize these results to a more general space-time $(\man, g_\text{WH+})$ with $\man = \mathbb{R}^2 \times K$ and metric
\begin{equation}
g_\text{WH+} = - \total t^2 + \mathe^{2 \alpha(r)} \total r^2 + \mathe^{2 \beta(r)} g_K \eqend{,}
\end{equation}
where $\alpha$ and $\beta$ are arbitrary smooth functions, and $(K,g_K)$ is a homogeneous Riemannian manifold. We obtain
\begin{proposition}
\label{prop:selfadjoint-whgen}
The spatial part of the Klein--Gordon operator for the spacetime $(\man, g_\text{WH+})$ (with a potential $V$ satisfying the conditions given in Sec.~\ref{sec:kg}) is essentially self-adjoint on $C_0^\infty(\mathbb{R} \times K)$, and the generalized wormhole spacetime $(\man, g_\text{WH+})$ is globally hyperbolic if
\begin{equation*}
\int_0^R \mathe^{\alpha(r)} \total r = \infty \quad \text{and} \quad \int_R^\infty \mathe^{\alpha(r)} \total r = \infty
\end{equation*}
for some $R > 0$.
\end{proposition}
\begin{proof}
The proof follows exactly as in the alternative proof of Lemma~\ref{lemma:wh-globalhyperbolic} (where now $\alpha$ and $\beta$ are general, and we employ Prop.~\ref{prop:warpedhom} instead of Prop.~\ref{prop:warpedsphere}) and in the proof of Prop.~\ref{prop:selfadjoint-wh} (where now $h$ is general).
\end{proof}

\section{Spherically symmetric, static Spacetimes}
\label{sec:spherical}

The most general form of the metric for a static, spherically symmetric spacetime $(\man, g_\text{ssym})$ reads
\begin{equation}
\label{met:sphsym}
g_\text{ssym} = - \mathe^{2 \alpha(r)} \total t^2 + \mathe^{2 \beta(r)} \total r^2 + r^2 g_{\mathbb{S}^2} \eqend{,}
\end{equation}
where $\alpha$ and $\beta$ are smooth functions, $t \in \mathbb{R}$, and $r \in (r_0,\infty)$. We have
\begin{proposition}
\label{prop:selfadjoint-ssym}
The spatial part of the Klein--Gordon operator for the spacetime $(\man,g_\text{ssym})$ (with a potential $V$ satisfying the conditions given in Sec.~\ref{sec:kg}) is essentially self-adjoint on $C_0^\infty(\Sigma)$, where $\Sigma = (r_0,\infty) \times \mathbb{S}^2$, if
\begin{equation}
\label{eq:ssym-condition}
\int_{r_0}^R \mathe^{\beta(r) - \alpha(r)} \total r = \infty \quad \text{and} \quad \int_R^\infty \mathe^{\beta(r) - \alpha(r)} \total r = \infty
\end{equation}
for some $R > r_0$.
\end{proposition}
\begin{proof}
The metric~\eqref{met:sphsym} has the form~\eqref{eq:foliation} with the lapse $N = \mathe^{\alpha(r)}$ and the spatial metric $h = \mathe^{2 \beta(r)} \total r^2 + r^2 g_{\mathbb{S}^2}$. Using~\cite[Theorem 4.1]{Much:2018ikx}, the conclusion thus follows if the Riemannian manifold $(\Sigma, \tilde{h} = N^{-2} h)$ is geodesically complete. By the Hopf--Rinow theorem, this holds if $(\Sigma, \mathe^{- 2 \alpha(r)} h)$ is a complete Riemannian manifold, which in turn follows from Prop.~\ref{prop:warpedsphere} with $a = r_0$, $b = \infty$ and replacing the function $\alpha \to \beta - \alpha$.
\end{proof}
\begin{corollary}
\label{corr:ssym-globalhyperbolic}
Under the conditions of Prop.~\ref{prop:selfadjoint-ssym}, $(\man,g_\text{ssym})$ is globally hyperbolic.
\end{corollary}
\begin{proof}
Combining Prop.~\ref{prop:conformal} (with the conformal factor $\Omega(r) = \mathe^{- 2 \alpha(r)}$) and Prop.~\ref{prop:warpedproduct1}, we conclude using the completeness of $(\Sigma, \mathe^{- 2 \alpha(r)} h)$ that we showed in the proof of Prop.~\ref{prop:selfadjoint-ssym}.
\end{proof}

We now consider various physically relevant examples.
\begin{example}
The metric of the exterior Schwarzschild spacetime $(\man,g_\text{BH})$ reads
\begin{equation*}
g_\text{BH} = - \left( 1 - \frac{2 M}{r} \right) \total t^2 + \left( 1 - \frac{2 M}{r} \right)^{-1} \total r^2 + r^2 g_{\mathbb{S}^2} \eqend{,}
\end{equation*}
where $t \in \mathbb{R}$ and $r \in (2M,\infty)$.
\begin{lemma}
\label{lemma:schwarzschild}
The exterior Schwarzschild spacetime $(\man,g_\text{BH})$ is globally hyperbolic, and the spatial part of the Klein--Gordon operator (with a potential $V$ satisfying the conditions given in Sec.~\ref{sec:kg}) is essentially self-adjoint on $C_0^\infty\left( (2 M,\infty) \times \mathbb{S}^2 \right)$.
\end{lemma}
\begin{proof}
We employ Prop.~\ref{prop:selfadjoint-ssym} and Corr.~\ref{corr:ssym-globalhyperbolic} and have to verify the condition~\eqref{eq:ssym-condition} with $\alpha(r) = - \beta(r) = \frac{1}{2} \ln\left( 1 - \frac{2 M}{r} \right)$, which are smooth functions for $r > 2 M$. That is, we need to determine $R > 2 M$ such that
\begin{equation*}
\int_{2 M}^R \left( 1 - \frac{2 M}{r} \right)^{-1} \total r = \infty \quad \text{and} \quad \int_R^\infty \left( 1 - \frac{2 M}{r} \right)^{-1} \total r = \infty \eqend{.}
\end{equation*}
Using the indefinite integral $\int \left( 1 - \frac{2 M}{r} \right)^{-1} \total r = r + 2 M \ln(r-2 M)$, the condition holds for any $R > 2 M$, and we conclude.
\end{proof}
\end{example}
We note that the self-adjointness of the spatial part of the Klein--Gordon operator for the exterior Schwarzschild spacetime has also been proven in~\cite{Dimock:1987hi,Dimock:1986ar,Kay:1988mu} using different methods.

\begin{example}
A static, spherically symmetric solution to the Einstein equations describing a charged black hole is given by the Reissner--Nordström Spacetime $(\man,g_\text{RN})$ with metric
\begin{equation*}
g_\text{RN} = - \left( 1 - \frac{2 M}{r} + \frac{Q^2}{r^2} \right) \total t^2 + \left( 1 - \frac{2 M}{r} + \frac{Q^2}{r^2} \right)^{-1} \total r^2 + r^2 g_{\mathbb{S}^2} \eqend{.}
\end{equation*}
Similar to the Schwarzschild spacetime, the region outside the outer event horizon $r > M + \sqrt{ M^2 - Q^2 }$ is globally hyperbolic and the spatial part of the Klein--Gordon operator is essentially self-adjoint on $C_0^\infty\left( (M + \sqrt{ M^2 - Q^2 },\infty) \times \mathbb{S}^2 \right)$. The proof is analogous to the Schwarzschild example Lemma~\ref{lemma:schwarzschild}.
\end{example}

\begin{example}
The de Sitter spacetime $(\man, g_\text{dS})$ is a maximally symmetric solution of the Einstein equations with positive cosmological constant $\Lambda > 0$. The static patch has the metric
\begin{equation*}
g_\text{dS} = - \left( 1 - \frac{r^2}{L^2} \right) \total t^2 + \left( 1 - \frac{r^2}{L^2} \right)^{-1} \total r^2 + r^2 g_{\mathbb{S}^2} \eqend{,}
\end{equation*}
where $t \in \mathbb{R}$, $r \in [0,L)$ and $L = \sqrt{\Lambda/3}$. This region is globally hyperbolic and the spatial part of the Klein--Gordon operator is essentially self-adjoint on $C_0^\infty\left( [0,L) \times \mathbb{S}^2 \right)$, and the proof is analogous to the previous examples. The surface $r = L$ is known as the cosmological horizon.
\end{example}

\begin{example}
The topological black hole~\cite{Vanzo:1997gw} has metric
\begin{equation*}
g_\text{TBH} = - V(r) \total t + \frac{1}{V(r)} \total r^2 + r^2 g_K \quad\text{with}\quad V(r) = \kappa - \frac{2 M}{r} - \frac{3}{\Lambda} r^2 \eqend{,}
\end{equation*}
where $t \in \mathbb{R}$, $r \geq 0$, and $g_K$ is the metric of a homogeneous space with Ricci tensor $R = \kappa g_K$. It is a solution of the Einstein equations with cosmological constant $\Lambda$, and we assume that $\Lambda = 3/L^2 > 0$. By a constant rescaling of the coordinates, we can furthermore fix $\kappa \in \{-1,0,1\}$. For $\kappa = 1$, we have $K = \mathbb{S}^2$ and obtain the Schwarzschild--de Sitter solution if the parameters satisfy $27 M^2 < L^2$. In this case, there are three real roots of the equation $V(r) = 0$ which read
\begin{equation*}
r_k = \frac{2 L}{\sqrt{3}} \cos\left[ \frac{(2k-1) \pi}{3} + \frac{1}{3} \arccos\left( \frac{3 \sqrt{3} M}{L} \right) \right] \eqend{,} \quad k \in \{0,1,2\} \eqend{,}
\end{equation*}
and we have $r_0 \geq r_1 \geq 0 \geq r_2$. The Schwarzschild--de Sitter black hole is globally hyperbolic if we restrict $r$ to the region $r \in (r_1,r_0)$ (between the event horizon at $r = r_1$ and the cosmological horizon at $r = r_0$), the spatial part of the Klein--Gordon operator is essentially self-adjoint on $C_0^\infty\left( (r_1,r_0) \times \mathbb{S}^2 \right)$, and the proof is analogous to the previous examples.
\end{example}
We note that all of the above also holds in dimensions $d > 4$, where the two-dimensional sphere $\mathbb{S}^2$ or homogeneous manifold $K$ is replaced by a $(d-2)$-dimensional one.

Lastly, we present a negative example:
\begin{example}
The Anti-de Sitter spacetime $(\man,g_\text{AdS})$ is a maximally symmetric solution of the Einstein equations with negative cosmological constant $\Lambda < 0$. The topology of Anti-de Sitter spacetime is $\man = \mathbb{S}^1 \times (0,\infty) \times \mathbb{S}^2$, and the spacetime contains closed timelike curves. However, in physics one usually works with the universal cover $\man_\infty = \mathbb{R} \times (0,\infty) \times \mathbb{S}^2$ and the metric
\begin{equation*}
g_\text{AdS} = - \left( 1 + \frac{r^2}{L^2} \right) \total t^2 + \left( 1 + \frac{r^2}{L^2} \right)^{-1} \total r^2 + r^2 g_{\mathbb{S}^2} \eqend{,}
\end{equation*}
where $L = \sqrt{-\Lambda/3}$. Nevertheless, even the universal cover is not globally hyperbolic due to its timelike boundary at spatial infinity $r \to \infty$, which allows light signals to travel to infinity and back in finite time. This can be checked using Corr.~\ref{corr:ssym-globalhyperbolic}, since the condition~\eqref{eq:ssym-condition} cannot be fulfilled: we have $\alpha(r) = - \beta(r) = \frac{1}{2} \ln\left( 1 + \frac{r^2}{L^2} \right)$ and
\begin{equation*}
\int_0^\infty \left( 1 + \frac{r^2}{L^2} \right)^{-2} \total r = \frac{\pi L}{4} < \infty \eqend{.}
\end{equation*}
Again, the same conclusion is obtained in all dimensions $d \geq 2$.
\end{example}

We summarize: a sufficient and explicit condition for static, spherically symmetric spacetimes to be globally hyperbolic and for the spatial part of the Klein--Gordon operator to be essentially self-adjoint is the fulfillment of Eq.~\eqref{eq:ssym-condition}, given that $\alpha$ and $\beta$ are smooth functions.

\section{Non-static spacetimes}
\label{sec:nonstatic}

The simplest and at the same time important examples of non-static spacetimes are the ones that are needed in cosmology:
\begin{example}
The Friedmann--Lema{\^\i}tre--Robertson---Walker (FLRW) spacetimes $(\man,g)$ are warped products $\man = (b,c) \times K$ with metric
\begin{equation*}
g = - \total t^2 + a(t)^2 g_K \eqend{,}
\end{equation*}
where $(K,g_K)$ is a homogeneous space with Ricci tensor $R = k g_K$ for a constant $k$, and $a \neq 0$ (the scale factor) is a smooth function.
\end{example}
\begin{lemma}
The FLRW spacetimes are globally hyperbolic, and the spatial part of the Klein--Gordon operator (with a potential $V$ satisfying the conditions given in Sec.~\ref{sec:kg}) is essentially self-adjoint on $C_0^\infty(K)$.
\end{lemma}
\begin{proof}
Global hyperbolicity follows from Prop.~\ref{prop:warpedproduct1} with $a \to b$, $b \to c$ and $f = a^2$, since $(K,g_K)$ is complete by Lemma~\ref{lemma:hom}. The FLRW metric has the form~\eqref{eq:foliation} with the lapse $N = 1$ and the spatial metric $h = a(t)^2 g_K$. Using~\cite[Theorem 4.1]{Much:2018ikx}, the conclusion thus follows if the Riemannian manifold $(K,h)$ is geodesically complete for each fixed $t$. Since $(K,g_K)$ is complete and for each fixed $t$ the spatial metric $h$ is a rescaling of $g_K$, thus inducing an equivalent norm, it follows that $(K,h)$ is complete and we conclude by the Hopf--Rinow theorem.
\end{proof}

\begin{example}
The type I Bianchi universes $(\man = (b,c) \times \mathbb{R}^3, g_\text{BI})$ are the simplest generalizations of the flat FLRW spacetimes if one does not assume isotropy. Including a distinct scale factor for each Cartesian direction, their metric reads
\begin{equation*}
g_\text{BI} = - \total t^2 + a_x(t)^2 \total x^2 + a_y(t)^2 \total y^2 + a_z(t)^2 \total z^2 \eqend{,}
\end{equation*}
where $a_x,a_y,a_z \neq 0$ are smooth functions.
\end{example}
\begin{lemma}
\label{lemma:bianchi1}
Type I Bianchi universes are globally hyperbolic. The spatial part of the Klein--Gordon operator (with a potential $V$ satisfying the conditions given in Sec.~\ref{sec:kg}) is essentially self-adjoint on $C_0^\infty(\mathbb{R}^3)$.
\end{lemma}
\begin{proof}
Global hyperbolicity follows from Prop.~\ref{prop:nonstatic-fourdim} with $h = 0$, $k = \ln a_x^2$, $l = \ln a_y^2$ and $m = \ln a_z^2$. The metric $g_\text{BI}$ has the form~\eqref{eq:foliation} with the lapse $N = 1$ and the spatial metric $h = a_x(t)^2 \total x^2 + a_y(t)^2 \total y^2 + a_z(t)^2 \total z^2$. Using~\cite[Theorem 4.1]{Much:2018ikx}, essential self-adjointness thus follows if the Riemannian manifold $(\mathbb{R}^3,h)$ is geodesically complete for each fixed $t$, which by the Hopf--Rinow theorem is equivalent to completeness. Since we have the bound
\begin{equation*}
\min( a_x^2(t), a_y^2(t), a_z^2(t) ) \eta \leq h \leq \max( a_x^2(t), a_y^2(t), a_z^2(t) ) \eta
\end{equation*}
with the flat metric $\eta = \total x^2 + \total y^2 + \total z^2$, we can apply Prop.~\ref{prop:bounded} (with $g_0 = \eta$, $a \to \min( a_x^2(t), a_y^2(t), a_z^2(t) )$ and $b \to \max( a_x^2(t), a_y^2(t), a_z^2(t) )$) to conclude, using that $(\mathbb{R}^3,\eta)$ is complete by Lemma~\ref{lemma:hom}.
\end{proof}
A special case of a Bianchi type I universe is the Kasner spacetime~\cite{kasner1921} with $a_x(t) = t^{p_1}$, $a_y(t) = t^{p_2}$ and $a_z(t) = t^{p_3}$, where the parameters $p_i$ fulfill the two conditions $p_1 + p_2 + p_3 = 1$ and $p_1^2 + p_2^2 + p_3^2 = 1$ and $t > 0$. In terms of the Lifshitz--Khalatnikov parameter $s \in [0,1]$, we may write~\cite{lifshitzkhalatnikov1961}
\begin{equation*}
p_1 = - \frac{s}{1+s+s^2} \eqend{,} \quad p_2 = \frac{1+s}{1+s+s^2} \eqend{,} \quad p_3 = \frac{s(1+s)}{1+s+s^2} \eqend{.}
\end{equation*}
The Kasner metric is significant because it approximates a general class of solutions near a cosmological singularity, see Ref.~\cite{wainwrightkrasinski2008} for details and references. By Lemma~\ref{lemma:bianchi1}, it is globally hyperbolic and the spatial part of the Klein--Gordon operator is self-adjoint (recalling that $t > 0$).

Lastly, we may also consider chaotic evolution:
\begin{example}
The type IX Bianchi or mixmaster~\cite{Misner:1969hg} universes $(\man,g_\text{BIX})$ have the metric
\begin{equation}
g_\text{BIX} = - \total t^2 + a(t)^2 \left[ \mathe^{\sqrt{3} \beta_-(t) + \beta_+(t)} \sigma_x^2 + \mathe^{- \sqrt{3} \beta_-(t) + \beta_+(t)} \sigma_y^2 + \mathe^{- 2 \beta_+(t)} \sigma_z^2 \right]
\end{equation}
with the one-forms
\begin{subequations}
\begin{align}
\begin{split}
\sigma_x &= \sin \theta \cos \phi \total \psi + \sin \psi \bigl( \cos \psi \cos \theta \cos \phi + \sin \psi \sin \phi \bigr) \total \theta \\
&\quad+ \sin \psi \sin \theta \bigl( \sin \psi \cos \theta \cos \phi - \cos \psi \sin \phi \bigr) \total \phi \eqend{,}
\end{split} \\
\begin{split}
\sigma_y &= - \sin \theta \sin \phi \total \psi + \sin \psi \bigl( - \cos \psi \cos \theta \sin \phi + \sin \psi \cos \phi \bigr) \total \theta \\
&\quad- \sin \psi \sin \theta \bigl( \sin \psi \cos \theta \sin \phi + \cos \psi \cos \phi \bigr) \total \phi \eqend{,}
\end{split} \\
\sigma_z &= \cos \theta \total \psi - \cos \psi \sin \psi \sin \theta \total \theta - \sin^2 \psi \sin^2 \theta \total \phi \eqend{,}
\end{align}
\end{subequations}
where $a$ (the scale factor) and $\beta_\pm$ (the shape parameters) are smooth functions of $t$, $\psi, \theta \in [0,\pi)$ and $\phi \in [0,2\pi)$. Here, we have chosen to use the usual hyperspherical coordinates~\cite[App.~B]{Taub:1950ez} instead of the Euler angles that are prevalent in Bianchi IX cosmology and which stem from the SU(2) approach to the symmetries of the three-sphere~\cite{King:1991jd}. While the algebraic symmetry relations
\begin{equation}
\total \sigma_x = \sigma_y \wedge \sigma_z \eqend{,} \quad \total \sigma_y = \sigma_z \wedge \sigma_x \eqend{,} \quad \total \sigma_z = \sigma_x \wedge \sigma_y
\end{equation}
take longer to verify in these coordinates, they have the advantage that the isotropic limit $\beta_\pm \to 0$ directly leads to the closed FLRW metric:
\begin{equation*}
\lim_{\beta_\pm \to 0} g_\text{BIX} = - \total t^2 + a(t)^2 \left[ \total \psi^2 + \sin^2 \psi \left( \total \theta^2 + \sin^2 \theta \total \phi^3 \right) \right] = - \total t^2 + a(t)^2 g_{\mathbb{S}^3} \eqend{.}
\end{equation*}
\end{example}
\begin{lemma}
\label{lemma:bianchi9}
Mixmaster universes are globally hyperbolic. The spatial part of the Klein--Gordon operator (with a potential $V$ satisfying the conditions given in Sec.~\ref{sec:kg}) is essentially self-adjoint on $C_0^\infty(\mathbb{R}^3)$.
\end{lemma}
\begin{proof}
The proof proceeds analogously to the one of Lemma~\ref{lemma:bianchi1}, where in the last step we use that instead of $(\mathbb{R}^3,\eta)$, also $(\mathbb{S}^3, g_{\mathbb{S}^3})$ is complete by Lemma~\ref{lemma:hom}.
\end{proof}

\appendix

\section*{Acknowledgements}

\parindent=0pt
This work was supported and funded by \textit{Kuwait University}, Research Project No. SM06/21.

\bibliographystyle{elsarticle-num}
\bibliography{newGH.bib}
	
\end{document}